\documentclass{llncs}

\usepackage{amsmath}
\usepackage{latexsym}
\usepackage{amsfonts,amssymb}
\usepackage{graphicx, gastex}
\usepackage{cite}
\usepackage{mathrsfs}
 
\DeclareMathOperator{\Mod}{ mod}
\DeclareSymbolFont{rsfscript}{OMS}{rsfs}{m}{n}
\DeclareSymbolFontAlphabet{\mathrsfs}{rsfscript}

\begin{document}

\title{Slowly Synchronizing Automata with Zero\\ and Incomplete Sets}

\author{Elena V. Pribavkina}

\institute{Ural State University, 620083 Ekaterinburg, Russia\\
\email{elena.pribavkina@usu.ru}}

\maketitle

\begin{abstract} Using combinatorial properties of incomplete sets in a free monoid we construct a series of $n$-state deterministic automata with zero whose shortest synchronizing word has length
$\frac{n^2}4+\frac{n}2-1$.

\end{abstract}

\section{Introduction}
Recall that a deterministic finite automaton
 $\mathrsfs{A}=\langle Q,A,\delta\rangle$ is defined by specifying a finite \emph{state set} $Q$, an \emph{input alphabet} $A$, and a \emph{transition function} $\delta: Q\times A\to
Q$. The function $\delta$ naturally
extends to the free monoid $A^{*}$; this extension is still
denoted by $\delta$.
An automaton $\mathrsfs{A}=\langle Q,A,\delta\rangle$ is said to be
\emph{synchronizing} (or \emph{reset}) if there is a \emph{synchronizing} (\emph{reset}) word, that is a word $w\in A^*$ which takes all the states of $\mathrsfs A$
to a particular one: $\delta(q,w)=\delta(q',w)$ for all $q,q'\in
Q$.

Reset automata turn out to have various applications in different fields such as model-based testing of reactive systems, robotics, dna-computing, symbolic dynamics. In view of the applications an important question is about the length of the \emph{shortest} reset word for a given synchronizing automaton. This issue has been widely studied over the past forty years,
especially in connection with the famous \v{C}ern\'{y} conjecture \cite{Ce64} which states that any $n$-state synchronizing automaton possesses a synchronizing word of length at most $(n-1)^2$.
 This conjecture has been proved for a large number of classes of synchronizing automata, nevertheless in general it remains one of the most longstanding open problems in automata theory. For more details see the surveys \cite{MaSa99,Sa05,Vo_Survey}. It is known (see for example \cite[Proposition~3]{Vo_CIAA07}) that the proof of the \v{C}ern\'{y}'s conjecture reduces to proving it in two particular cases: for reset automata whose underlying graph is strongly-connected, i.\,e.\ each state is reachable from any other one, and for reset automata with \emph{zero}, i.\,e.\ with a particular state $0$ such that  $\delta(0,a)=0$ for any $a\in A$. Thus obtaining bounds for the maximal length of shortest synchronizing words for the class of \emph{reset automata with zero} is a rather natural and interesting problem.

 It is clear that any synchronizing automaton with zero possesses a unique zero state, and any synchronizing word brings the automaton in the zero state. A rather simple argument shows that the length of a synchronizing word for a given $n$-state reset automaton with zero is at most $\frac{n(n-1)}2$, see e.\,g. \cite{Ryst}. This bound is tight, since for each $n$ there is an $n$-state reset automaton with zero and $n-1$ input letters whose shortest reset word has length $\frac{n(n-1)}2$ (Fig.~\ref{fig_mart}).

 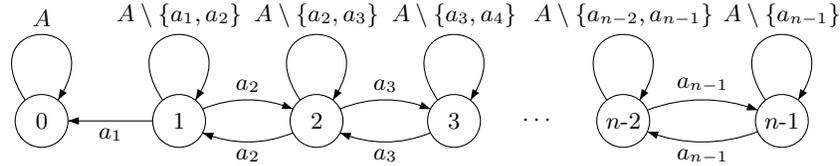
\begin{figure}[ht]
 \label{fig_mart}
\begin{center}
  \unitlength=4pt
    \begin{picture}(70,15)(1,-5)
    \gasset{Nw=5,Nh=5,Nmr=2.5}
    \thinlines
    \node(A0)(0,0){$0$}
    \node(A1)(13,0){$1$}
    \node(A2)(26,0){$2$}
    \node(A3)(39,0){$3$}
    \node[Nframe=n](A4)(47,0){$\cdots$}
    \node(A5)(55,0){$n$-$2$}
    \node(A6)(70,0){$n$-$1$}
    \drawedge(A1,A0){$a_1$}
    \drawedge[curvedepth=2](A2,A1){$a_2$}
    \drawedge[curvedepth=2](A1,A2){$a_2$}
    \drawedge[curvedepth=2](A2,A3){$a_3$}
    \drawedge[curvedepth=2](A3,A2){$a_3$}
    \drawedge[curvedepth=2](A5,A6){$a_{n-1}$}
    \drawedge[curvedepth=2](A6,A5){$a_{n-1}$}
    \drawloop(A0){$A$}
    \drawloop(A1){$A\setminus\{a_1,a_2\}$}
    \drawloop(A2){$A\setminus\{a_2,a_3\}$}
    \drawloop(A3){$A\setminus\{a_3,a_4\}$}
    \drawloop(A5){$A\setminus\{a_{n-2},a_{n-1}\}$}
    \drawloop(A6){$A\setminus\{a_{n-1}\}$}
    \end{picture}
 \end{center}  \caption{An $n$-state reset automaton with zero over an $n-1$-lettered alphabet whose shortest reset word is of length $\frac{n(n-1)}{2}$.}
\end{figure}

An essential feature of the example in Fig.~\ref{fig_mart} is that the input alphabet
size grows with the number of states. This contrasts with the aforementioned
examples due to \v{C}ern\'{y} \cite{Ce64} in which the alphabet does not depend on
the state number. Thus a natural question is to determine
the maximum length $c_m(n)$ of shortest reset words for $n$-state synchronizing
automata with zero over a fixed $m$-lettered input alphabet as a function of $n$.

Recently in \cite{Mart} with the help of computer experiments P. Martjugin found a series of  $n$-state automata with zero over a binary alphabet
whose shortest reset words have length $\frac{n^2}4+o(n^2)$. The main result of \cite{Mart} can be stated as follows:

\begin{theorem}
\label{mart} For each integer $n\ge8$ there exists a synchronizing $n$-state automaton with zero over a binary alphabet
whose shortest synchronizing word has length
$\left\lceil\frac{n^2+6n-16}{4}\right\rceil.$
\end{theorem}

Note that the construction from \cite{Mart} is not trivial and beside that, it can not be extended on larger alphabets. Let us explain what we mean by such an extension. We say that a synchronizing automaton $\mathrsfs{A}=\langle Q, A,
\delta\rangle$ is \emph{proper}, if each letter of the alphabet $A$
appears in each word synchronizing this automaton. Putting this another way, each letter is essential for synchronization of such an automaton. Naturally, it is the class of proper automata for which the problem of estimation of the function $c_m(n)$ should be considered, but adding new letters to the automata from Theorem~\ref{mart} violates this property.

The main result of the present paper is the following

\begin{theorem}
\label{theorem_main} Let $A$ be an alphabet with $|A|\ge2$.
For any integer $k>|A|$ there is a proper synchronizing automaton with zero and
 $n=2k$ states over the alphabet $A$ whose shortest synchronizing word has length
$$\frac{n^2}4+\frac{n}2-1.$$
\end{theorem}

Our construction leads to the same growth rate of the length of the shortest reset word as the construction from~\cite{Mart}, but essentially differs from it since there are no limitations on the number of input letters. Another important feature is that is was found not by a brute force, but as a result of analysis of interrelations between synchronizing automata with zero and incomplete sets in free monoids. Recall that complete and incomplete sets play a significant role in combinatorics on words and theory of codes in connection with the notion of a maximal code (see e.\,g.\ \cite{Rest}). We think that the relation between combinatorial objects of different nature (codes and automata) that we establish in the present paper is of a self-dependent interest. Recently such a relation has been independently discovered by Rampersad and Shallit in \cite{RamSha09} where the computational complexity of some universality problems is studied. In particular the authors obtained a result similar to our Proposition \ref{incompl_synchr} and a partial result of our Proposition \ref{incompl_short}.

The paper is organized as follows. In Section~\ref{incomplete} we introduce the notion of an incomplete set
and study properties of incompletable words necessary for the proof of the main result. In Section~\ref{semiflower} from a finite set of words
$X$ we build an automaton $\widehat{\mathrsfs{F}}(X)$ with zero, recognizing the monoid $X^*$. In Section~\ref{main} we establish the equality between words incompletable in $X^*$, and words synchronizing the constructed automaton $\widehat{\mathrsfs{F}}(X)$, which is used for the proof of the main Theorem~\ref{theorem_main}.

\section{Incomplete Sets}
\label{incomplete}

To fix the notation let us recall the main definitions from combinatorics on words.
By $|w|$ we denote the \emph{length} of the word $w$, the length of the \emph{empty word} $\lambda$ is equal to zero: $|\lambda|=0$.
By $A^+$ we denote the set of all non-empty words over the alphabet $A$,
and by $A^k$ -- the set of all words of length $k$ over $A$.
A word $u\in A^+$ is a \emph{factor} of $w$
(\emph{prefix} or \emph{suffix} respectively), if $w$ can be decomposed as
 $w=xuy$ ($w=uy$ or $w=xu$ respectively)
for some $x,y\in A^*$. A factor (prefix, suffix) $u$ of
$w$ is called \emph{proper} if $u\ne w.$
Given a word $u=a_1a_2\cdots a_n\in A^+$ by $u[i\ldots j]$ with
$1\le i,j\le n$ we denote the factor $a_ia_{i+1}\cdots a_{j}$
if $i\le j$, and the empty word if $i>j$. Moreover, we put $u[0]=\lambda$.
A word $u\in A^*$ is called \emph{unbordered} if none of its proper prefixes is its suffix.

Let $X$ be a finite set of words over the alphabet $A$. A word
$w\in A^*$ is said to be \emph{completable} in $X^*$ if $w$ is a factor
of some word in $X^*$ (to put this another way, $w$ can be ``covered" by elements from $X$), otherwise $w$ is \emph{incompletable}. We say that the set $X$ is
\emph{complete} if any word over $A$ is completable in $X^*$, otherwise the set $X$ is said to be
\emph{incomplete}.
Some properties of compete and incomplete sets were studied in \cite{Rest}.
Here we use the following result from \cite{Rest}:

\begin{proposition}
\label{restivo}
 Let $X\subseteq A^*$, $k=\max\limits_{x\in X}|x|$, and there is a word $u$ of length $k$ such that no element of $X$ is a factor of $u$. Then the set $X$ is incomplete, and the word
$w=(ua)^{k-1}u$ for an arbitrary letter $a$ is incompletable.
\end{proposition}

Next we consider the sets of the form $X=A^k\setminus\{u\}$,
where $u$ is some unbordered word and $k\ge 2$. The Proposition~\ref{restivo} implies the following

\begin{corollary}
\label{cor_rest} Any set of the form $X=A^k\setminus\{u\}$ with
$k\ge2$ and $u\in A^k$ is incomplete.

\end{corollary}

\begin{lemma}
\label{incompl_form} Let $X=A^k\setminus\{u\}$ for some unbordered word word $u$. Any incompletable word in $X^*$ is of the form
$v_0uv_1u\cdots uv_muv_{m+1}$ with $v_i\in A^*\setminus A^*uA^*$,
$m>0$.

\end{lemma}

\begin{proof}
Let $w$ be an incompletable word in $X^*$, then $u$ appears as a factor in $w$.
Indeed, if it is not the case, consider the shortest word $z\in A^*$ such that the length of the word $wz$ is a multiple of $k$. Then $wz$ can be decomposed as $wz=x_1x_2\cdots x_s$ with $x_1,x_2,\ldots,x_s\in A^k$ (Fig.~2). Since $w$ does not contain $u$ as a factor
for all $i$ we have $x_i\ne u$, hence $x_1,x_2,\ldots,x_s\in
A^k\setminus\{u\}$ and the word $w$ is completable in $X^*$, which is a contradiction.

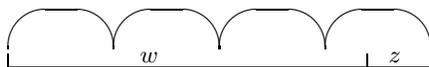
\begin{figure}[ht]
\begin{center}
\label{fig_inc}
  \unitlength=1pt
 \begin{picture}(160,15)(0,10)
 \put(0,0){\line(1,0){160}}
 \put(50,2){$w$}
 \put(136,0){\line(0,1){5}}
 \put(144,2){$z$}
 \put(0,0){\line(0,1){5}}
 \put(160,0){\line(0,1){5}}
 \put(20,7){\oval(40,30)[t]}
 \put(60,7){\oval(40,30)[t]}
 \put(100,7){\oval(40,30)[t]}
 \put(140,7){\oval(40,30)[t]}
 \end{picture}
 \end{center}
  \caption{The word $w$ does not contain $u$ as a factor.}
\end{figure}

Suppose now that $u$ appears just once as a factor of $w$, i.\,e.\ $w=w'uw''$ and $w',w''\in A^*\setminus A^*uA^*$. Then the word $u[2\ldots k]w''$ does not contain
$u$ as a factor since $u$ is unbordered. From the previous argument we deduce that there is a word $z\in A^*$ such that $u[2\ldots k]w''z\in X^*$.
In the same way, $w'u[1]$ does not contain $u$ as a factor, hence
there is a word $y\in A^*$ such that  $yw'u[1]\in X^*$ (Fig.~3). Thus $ywz\in X^*$, i.\,e.\ also in this case $w$ is completable in $X^*$. We come to a contradiction.

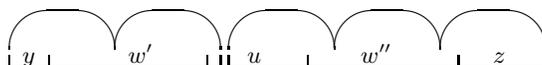
\begin{figure}[ht]
\begin{center}
\label{fig_3}
 \unitlength=1pt
 \begin{picture}(203,15)(0,10)
 \put(0,0){\line(1,0){80}}
 \put(83,0){\line(1,0){120}}
 \put(45,2){$w'$}
 \put(5,2){$y$}
 \put(15,0){\line(0,1){5}}
 \put(203,0){\line(0,1){5}}
 \put(75,0){\line(0,1){5}}
 \put(113,0){\line(0,1){5}}
 \put(183,2){$z$}
 \put(90,2){$u$}
 \put(133,2){$w''$}
 \put(0,0){\line(0,1){5}}
 \put(83,0){\line(0,1){5}}
 \put(80,0){\line(0,1){5}}
 \put(170,0){\line(0,1){5}}
 \put(20,7){\oval(40,30)[t]}
 \put(60,7){\oval(40,30)[t]}
 \put(103,7){\oval(40,30)[t]}
 \put(143,7){\oval(40,30)[t]}
 \put(183,7){\oval(40,30)[t]}
 \end{picture}

 \end{center}
  \caption{The word $w$ contains $u$ as a factor only once.}
\end{figure}

Therefore, the factor $u$ occurs in $w$ at least twice. Moreover, since  $u$ is unbordered these occurrences do not overlap, and  $w$ can be represented as $w=v_0uv_1u\cdots
v_muv_{m+1}$, where no $v_i$ contains $u$ as a factor, and $m>0$.
\end{proof}

Next we give a necessary and sufficient condition for a word of the form
 $w=v_0uv_1u\cdots v_muv_{m+1}$ to be incompletable. To this end we introduce some auxiliary notions.

A position $0\le i<k$ in the $j$-th occurrence of the word $u$ in $w$ is called
\emph{forbidden} if $u[i+1\ldots k]v_ju\cdots v_muv_{m+1}z$ does not belong
to $X^*$ for any $z\in A^*$. For the $j$-th occurrence of the word  $u$ in $w$, where $j$ ranges from $1$ to $m$, by $S_j$ we denote the \emph{set of forbidden positions}. Note that $S_j\subseteq\{0,1,\ldots k-1\}$ for all $j$'s.
A simple observation is the following

\begin{remark}\label{S1} A word $w=v_0uv_1u\cdots v_muv_{m+1}$ is incompletable in $X^*$ if and only if $S_1=\{0,1,\ldots, k-1\}$.
\end{remark}

A set $S=\{s_1,s_2\ldots,s_m\}\subseteq\mathbb{N}$ is called
 \emph{$k$-representative} if among the residues modulo $k$ of its elements there are all possible positive residues modulo $k$:
$$\{\bar{1},\bar{2},\ldots,\overline{k-1}\}\subseteq\{s_1\Mod k,s_2\Mod k,\ldots,s_m\Mod k\}.$$

Now we are ready to proof the criterion for a word $w$ to be incompletable.

\begin{lemma}
\label{incompl_crit}
Let $X=A^k\setminus\{u\}$ for some unbordered word $u$. A word $w=v_0uv_1uv_2\cdots
uv_muv_{m+1}$ with $v_i\in A^*\setminus A^*uA^*$, $m>0$ is incompletable in $X^*$ if and only if the set
$$\{|v_1|,\ |v_1|+|v_2|,\ldots,\ |v_1|+|v_2|+\cdots+|v_m|\}$$ is $k$-representative.
\end{lemma}

\begin{proof}
Consider a word $w=v_0uv_1uv_2\cdots uv_muv_{m+1}$. Since the set  $X$ contains only words of length $k$, for all $j$
we have $0\in S_j$. Moreover, since no proper prefix of $u$ is its suffix we have
$$S_{m+1}=\{0\}.$$ Further, it is easy to see that if the set
 $S_j$ is already defined, then a non-zero position $i$ belongs to $S_{j-1}$
if and only if $$u[i+1\ldots k]v_{j-1}u[1\ldots \ell]\in
X^*,\ \ell\in S_j.$$
Therefore $0\ne i\in S_m$ if and only if
$u[i+1\ldots k]v_m\in X^*$, i.\,e.\ the length of the factor $u[i+1\ldots
k]v_m$ is divisible by $k$, thus $i\equiv |v_m|\Mod k$. We get
$$S_m=\{0, |v_m|\Mod k\}.$$ Next suppose we have already calculated
\begin{equation}\label{Sj}
S_{j}=\{0,
|v_j|\Mod k, (|v_j|+|v_{j+1}|)\Mod k,\ldots, (|v_j|+\cdots+|v_m|)\Mod
k\}\end{equation}
for $0<j\le m$. Let us prove that
\begin{equation}
\label{Sj1}
\begin{array}{l}
  S_{j-1}=\{0, |v_{j-1}|\Mod k,
(|v_{j-1}|+|v_j|)\Mod k, \ldots, \\
\phantom{*}\hfill(|v_{j-1}|+|v_j|+\cdots+|v_m|)\Mod k\}.
 \end{array}
\end{equation}
Let $0<i<k$. Since $X=A^k\setminus\{u\}$, and
$u$ is unbordered, there is an integer $\ell$ such that $0\le
\ell<k$ and $u[i+1\ldots k]v_{j-1}u[1\ldots \ell]\in X^*$, i.\,e.\
the length of this factor is divisible by $k$, hence
$$i\equiv(|v_{j-1}|+\ell)\Mod k.$$ Since $0\ne i\in S_{j-1}$
if and only if $\ell\in S_{j}$, using \eqref{Sj} we obtain
получаем \eqref{Sj1}.

Therefore, $$S_1=\{0, |v_1|\Mod k, (|v_1|+|v_2|)\Mod k, \ldots,
(|v_1|+\cdots+|v_m|)\Mod k\}.$$ Taking into account Remark~\ref{S1} we conclude the proof.
\end{proof}

\begin{proposition}
\label{incompl_short}
 The shortest incompletable word for the set
$X=A^k\setminus\{u\}$, where $u$ is an unbordered word, has length
$k^2+k-1$.
\end{proposition}

\begin{proof} Proposition~\ref{restivo} (or Lemma~\ref{incompl_crit}) implies that the word
$w=(ua)^{k-1}u$ is incomletable in $X^*$ for any $a\in A$. Such a word has length $k^2+k-1$.

Let $w'$ be the shortest incompletable word for the set $X^*$.
Then by Lemma~\ref{incompl_form} it has the form $v_0uv_1\ldots
v_muv_{m+1}$ with $v_i\in A^*\setminus A^*uA^*$ and $m>0$.
Note that $m\ge k-1$ and among $v_i$'s ($1\le i\le m$) there should be at least $k-1$ non-empty words (otherwise the set
$\{|v_1|,|v_1|+|v_2|,\ldots, |v_1|+\cdots+|v_m|\}$ is not $k$-representative, and by Lemma~\ref{incompl_crit} the word $w$
can be completed). Thus, $|w'|\ge k-1+k^2=|w|.$

\section{Construction of the Automaton from the set $X$}
\label{semiflower}

In this Section we consider finite automata as devices for recognizing languages. Recall that to this end we choose an \emph{initial state} $q_0$ and a set $F$ of \emph{terminal states}. An automaton
$\mathrsfs{A}=\langle Q,A,\delta,q_0,F\rangle$ is said to \emph{recognize}
a language $L\subseteq A^*$ if
$$L=\{w\in A^*\mid \delta(q_0,w)\in F\}.$$

Let $X$ be a finite set of words. By $\mathrsfs{F}(X)$ we denote an automaton recognizing the monoid  $X^*$, in which the initial state  $1$ is also the only terminal one, all the cycles pass through this state, and moreover, all the words from $X$ label all possible simple cycles. Such automata are known as \emph{semi-flower automata}, see e.\,g.\ \cite{GiamRest}.
An example of a semi-flower automaton for the set
$X=\{aa,ab,ba,bb,aab\}$ is presented on Fig.~4.

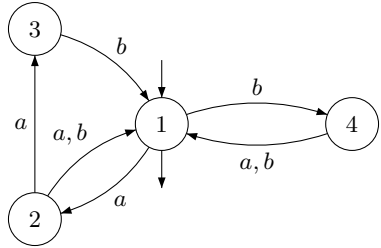
\begin{figure}[ht]
 \label{fig_flo1}
 \begin{center}
   \unitlength=4pt
    \begin{picture}(30,18)(0,0)
    \gasset{Nw=5,Nh=5,Nmr=2.5}
    \thinlines
    \node[Nmarks=if,iangle=90,fangle=-90](A1)(12,9){$1$}
    \node(A2)(0,0){$2$}
    \node(A3)(0,18){$3$}
    \node(A4)(30,9){$4$}
    \drawedge[curvedepth=2](A1,A2){$a$}
    \drawedge(A2,A3){$a$}
    \drawedge[curvedepth=2](A1,A4){$b$}
    \drawedge[curvedepth=2](A4,A1){$a,b$}
    \drawedge[curvedepth=2](A3,A1){$b$}
    \drawedge[curvedepth=2](A2,A1){$a,b$}
    \end{picture}
 \end{center}
  \caption{The semi-flower automaton for $X^*$, $X=\{aa,ab,ba,bb,aab\}$.}
\end{figure}

In general a semi-flower automaton is nondeterministic.
Recall that for such an automaton the deterministic transition function $\delta: Q\times A\to Q$ is replaced with the function $\delta: Q\times A\to 2^Q$. Note that this definition does not exclude the possibility
$\delta(q,a)=\varnothing$ for some $q\in Q$ and $a\in A$
(for instance, for the automaton in Fig.~4 we have
$\delta(3,a)=\varnothing$). This function as in the case of deterministic automata, can be extended to the free monoid $A^*$: if $w=au$, $a\in A$, $u\in A^+$ then we put
$$\delta(q,au)=\bigcup\limits_{t\in\delta(q,a)}\delta(t,u).$$

Any semi-flower automaton can be completed by adding a new state $0$ and putting all previously undefined transitions to be equal $0$ (Fig.~5). Moreover, put $\delta(0,a)=0$ for all $a\in A$. The automaton obtained in this way is denoted by $\widehat{\mathrsfs{F}}(X)$.

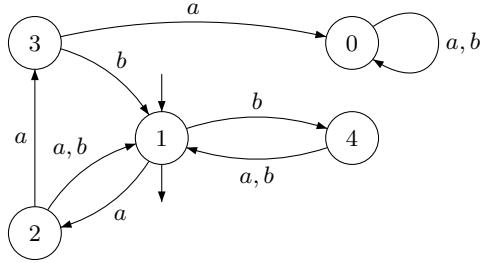
\begin{figure}[ht]
\label{fig_flo2}
 \begin{center}
\unitlength=4pt
    \begin{picture}(30,20)(0,0)
    \gasset{Nw=5,Nh=5,Nmr=2.5}
    \thinlines
    \node[Nmarks=if,iangle=90,fangle=-90](A1)(12,9){$1$}
    \node(A2)(0,0){$2$}
    \node(A3)(0,18){$3$}
    \node(A4)(30,9){$4$}
    \drawedge[curvedepth=2](A1,A2){$a$}
    \drawedge(A2,A3){$a$}
    \drawedge[curvedepth=2](A1,A4){$b$}
    \drawedge[curvedepth=2](A4,A1){$a,b$}
    \drawedge[curvedepth=2](A3,A1){$b$}
    \drawedge[curvedepth=2](A2,A1){$a,b$}
    \node(A0)(30,18){$0$}
    \drawedge[curvedepth=2](A3,A0){$a$}
    \drawloop[loopangle=0](A0){$a,b$}
    \end{picture}
  \end{center}
\caption{The automaton $\mathrsfs{\widehat{F}}(X)$ for
$X=\{aa,ab,ba,bb,aab\}$.}
\end{figure}

In particular case when $X=A^k\setminus\{u\}$, the automaton
$\mathrsfs{\widehat F}(X)$ will be denoted by
$\mathrsfs{\widehat{F}}(k,u)$.

\section{Incomplete Sets and Synchronizing Automata}
\label{main}

The notion of synchronization of an automaton can be extended for the case of nondeterministic automata in different ways. Here we use one of such extensions which is known as \emph{strong synchronization} \cite{Ito04}.

A nondeterministic finite automaton $\mathrsfs{A}=\langle
Q,A,\delta\rangle$ is said to be \emph{synchronizing} if there is a word $w\in A^*$ and a state $q\in Q$ such that
$\delta(q',w)=\{q\}$ for any $q'\in Q$. Putting this another way, all possible paths labeled by the word $w$ from an arbitrary state of $\mathrsfs{A}$ lead to the particular state $q$.

The following proposition connects the notions of a word incompletable in
$X^*$ and a synchronizing word for the completed semi-flower automaton $\mathrsfs{\widehat{F}}(X)$.

\begin{proposition}
\label{incompl_synchr}
Given an incomplete set $X$, the word $w$ is incompletable in  $X^*$ if and only if $w$ is synchronizing for
the automaton $\mathrsfs{\widehat{F}}(X)$.
\end{proposition}

\begin{proof}
By the definition of $\mathrsfs{F}(X)$ the fact that the word $w$ is not a factor of any word in $X^*$ means that it can not be read from any state of this automaton. Equivalently, this word brings to the state $0$ any state of the automaton $\mathrsfs{\widehat{F}}(X)$.
\end{proof}

Now we construct the automaton $\mathrsfs{\widehat{F}}(k,u)=\langle
Q,A,\delta,1,\{1\}\rangle$ for $k\ge2$  and the unbordered word
$u=a_1a_2\cdots a_k$.

This automaton has $n=2k$ states $Q=\{0,1,\ldots,2k-1\}$, and the transitions are defined as follows (see Fig. 6).

For $1\le i\le k-1$ we put

$\delta(i,a_i)=i+1;$

$\delta(i,b)=k+i$ for all $b\in A\setminus\{a_i\}$.

$\delta(k,a_k)=0$;

$\delta(k,b)=1$ for all $b\in A\setminus\{a_k\}$.

For $k+1\le i\le 2k-2$ we put

$\delta(i,a)=i+1$ for all $a\in A$;

$\delta(2k-1,a)=1$ for all $a\in A$;

$\delta(0,a)=0$ for all $a\in A$.

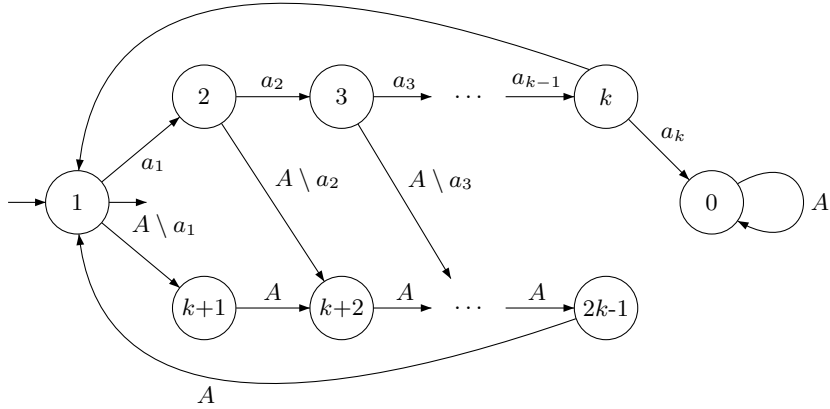
\begin{figure}[ht]
\label{fig_flodet}
 \begin{center}
  \unitlength=4pt
    \begin{picture}(60,35)(0,0)
    \gasset{Nw=6,Nh=6,Nmr=3}
    \thinlines
    \node[Nmarks=if](A1)(0,20){$1$}
    \node(A0)(60,20){$0$}
    \drawloop[loopangle=0](A0){$A$}
    \node(A2)(12,30){$2$}
    \node(AK1)(12,10){$k$+$1$}
    \node(A3)(25,30){$3$}
    \node(AK2)(25,10){$k$+$2$}
    \node[Nw=7,Nframe=n](A4)(37,30){$\ldots$}
    \node[Nw=7,Nframe=n](AK3)(37,10){$\ldots$}
    \node(AK)(50,30){$k$}
    \node(A2K1)(50,10){$2k$-$1$}
    \drawedge(AK,A0){$a_k$}
    \drawedge[ELside=r](A1,A2){$a_1$}
    \drawedge(A1,AK1){$A\setminus a_1$}
    \drawedge(A2,AK2){$A\setminus a_2$}
    \drawedge(A2,A3){$a_2$}
    \drawedge(AK1,AK2){$A$}
    \drawedge(A3,AK3){$A\setminus a_3$}
    \drawedge(A3,A4){$a_3$}
    \drawedge(AK2,AK3){$A$}
    \drawqbpedge(A2K1,200,A1,-90){$A$}
    \drawqbpedge[ELside=r,syo=2](AK,160,A1,90){$A\setminus a_k$}
    \drawedge(A4,AK){$a_{k-1}$}
    \drawedge(AK3,A2K1){$A$}
    \end{picture}

 \end{center}
  \caption{The automaton $\mathrsfs{\widehat{F}}(k,u)$ for $k\ge2$ and $u=a_1a_2\cdots a_k$.}
\end{figure}

Note that in this case the automaton $\mathrsfs{\widehat{F}}(k,u)$
is deterministic; moreover by proposition~\ref{incompl_synchr} the word $w=(ua)^{k-1}u$ is its shortest reset word. Therefore it holds the following

\begin{proposition}
\label{prop_main} Given an unbordered word $u$ of length $k\ge2$ and
$X=A^k\setminus\{u\}$, the shortest synchronizing word for the deterministic automaton $\mathrsfs{\widehat{F}}(k,u)$ with zero has length $k^2+k-1$.
\end{proposition}

\begin{proposition} \label{prop_proper} Let $k>|A|$, and let $u$ be an unbordered word of length $k$ containing all the letters of the alphabet $A$. Then the automaton $\mathrsfs{\widehat{F}}(k,u)$ is proper.
\end{proposition}

\begin{proof}
Lemma~\ref{incompl_form} implies that the word $u$ is a factor of any incompletable word for $A^k\setminus\{u\}$,
hence by Proposition~\ref{incompl_synchr} $u$ is a factor of any reset word for the automaton
$\mathrsfs{\widehat{F}}(k,u)$. Since every letter of the alphabet $A$ occurs in
 $u$, then every letter occurs in each reset word for
 $\mathrsfs{\widehat{F}}(k,u)$, thus by definition this automaton is proper.
\end{proof}

Propositions~\ref{prop_main} and  \ref{prop_proper} imply the main theorem~\ref{theorem_main}.

\end{proof}

\end{document}